\documentclass[letterpaper, 10 pt, conference]{ieeeconf}  
\IEEEoverridecommandlockouts 

\pdfoutput=1
\usepackage{amsmath,amssymb,amsfonts}
\usepackage{graphicx}
\usepackage{textcomp}
\usepackage{comment}
\usepackage[reftex]{theoremref}
\usepackage{color}

\usepackage{epstopdf}
\newtheorem{assumption}{Assumption}
\newtheorem{problem}{Problem}
\newtheorem{theorem}{Theorem}
\newtheorem{proposition}{Proposition}
\newtheorem{remark}{Remark}
\newtheorem{lemma}{Lemma}
\newtheorem{definition}{Definition}
\newtheorem{corollary}{Corollary}

\newtheorem{fact}{Fact}

\addtocounter{MaxMatrixCols}{10}

\begin{document}

\title{Learning controllers from data via kernel-based interpolation\\
\author{Zhongjie Hu, Claudio De Persis, Pietro Tesi
\thanks{Zhongjie Hu and Claudio De Persis are with
the Engineering and Technology Institute, University of Groningen, 9747AG, The Netherlands (e-mail: zhongjie.hu@rug.nl, c.de.persis@rug.nl). Pietro Tesi is with DINFO, University of Florence, 50139 Florence,
Italy (e-mail: pietro.tesi@unifi.it). *This publication is also part
of the project Digital Twin with project number P18-03 of
the research programme TTW Perspective which is (partly)
financed by the Dutch Research Council (NWO).}%
}
}

\maketitle

\begin{abstract}
We propose a data-driven control design method for nonlinear systems that builds on kernel-based interpolation. Under some assumptions on the system dynamics, kernel-based functions are built from data and a model of the system, along with deterministic model error bounds, is determined. Then, we derive a controller design method that aims at stabilizing the closed-loop system by cancelling out the system nonlinearities. The proposed method can be implemented using semidefinite programming and returns positively invariant sets for the closed-loop system.
\end{abstract}

\section{Introduction}
Data-driven control is a cornerstone of automatic control. Starting from the pioneering work by Ziegler--Nichols \cite{pid}, data-driven control has proved effective in contexts where finding a model of the system from first principles is difficult or time-consuming, and a controller is instead determined using experimental data. In the last years, there has been a renewed interest in data-driven control, and the reason is the growing complexity of the engineering systems for which first-principle laws are often difficult to determine.
 
The body of work on data-driven control is extremely vast, and it is not our goal to provide here any comprehensive  review. We will focus on the basic problem of designing a feedback controller and consider
batch (i.e., non-iterative) methods, that 
are methods in which a controller is computed once and for all using a finite set of data collected from the system. The interest for batch methods is related to the possibility of having \emph{finite-sample} stability guarantees, 
as opposed to classic adaptive control schemes that usually only provide asymptotic guarantees. 

\emph{Related work}. Batch methods can be classified as \emph{indirect} or \emph{direct}. In the first case, data are used to build a model of the system (within a selected model class, e.g. linear models). In this process, explicit error bounds arising from noise in the data or a mismatch between system and model class can also be determined. Then, model-based control design techniques are applied. In contrast, direct method go directly from data to the controller. Also direct methods can involve notions of model class and uncertainty but the decision variables are directly the controller parameters, without any intermediate identification step. 

Most of the existing works consider linear systems and assume that there are no unmodeled dynamics, which means that the plant-model mismatch is at most parametric. Recent contributions in this context are \cite{dean2019sample,umenberger2019robustLCSS}
for what concerns indirect methods and 
\cite{de2019formulas,berberich2020robust,henk-ddctr-uncer} for what concerns direct methods. Dealing with nonlinear systems is arguably much more difficult. One main reason is that it becomes harder to compute finite-sample uncertainty bounds, even when the uncertainty is purely parametric. Another main reason is that controller design for nonlinear systems is itself much more complex. 
Recent contributions that consider parametric uncertainty tackle bilinear systems \cite{bisoffi2020data,yuan2021data}, polynomial (and rational) systems \cite{guoTAC2021poly,dai2020semi,Nejati2022poly,strasser2021data}, and
LPV systems \cite{Verhoek2022lpv}.
For general nonlinear systems, but still in the context of parametric uncertainty, 
we find linearly parametrized models with known basis functions \cite{dai2021statedependent,NonlinearityCancellation2023}.
The result in \cite{NonlinearityCancellation2023}, in particular, introduces a controller design technique that provides, under rather mild conditions, finite-sample stability guarantees along with an estimate of regions of attraction and positive invariant sets for the closed-loop system. 

Assuming the exact knowledge of the basis functions is reasonable in many practical cases such as with
mechanical and electrical systems in which some prior information about the dynamics is available but the exact systems parameters may be unknown. In many other cases, however, this prior information may be unknown. 
Methods that consider this scenario include methods based on Gaussian process models 
\cite{Umlauft2018GP,devonport2020bayesian}, methods based on 
linear \cite{de2019formulas,Fraile2020FL,luppi2022data,Cheah2022NL,Cetinkaya2021} 
and polynomial approximations \cite{Guo2022Taylor,Martin2022Taylor}, 
and methods based on linearly parametrized models with partially known basis functions \cite{NonlinearityCancellation2023}.
Despite the differences, the common idea is to describe the system via a quantity which is known up to parametric uncertainty and treat unmodeled dynamics as an error term, i.e., a remainder. The challenge is thus twofold: (i) to derive finite-sample bounds for the remainder and (ii) to design a control law that is robust to the uncertainty that this remainder introduces. 

\emph{Contribution and outline of the paper}. In this paper, we consider the last scenario discussed above, that is the scenario where the system to control has general dynamics (e.g. not necessarily bilinear or polynomial) and there is no prior knowledge of the true basis functions. We propose a new method that combines ideas from kernel-based identification \cite{pillonetto2014kernel} and the controller design method introduced in \cite{NonlinearityCancellation2023}. Specifically, we consider an \emph{indirect} method that consists of two steps: we first determine a kernel-based model of the system along with deterministic error bounds, in line with recent results on kernel learning \cite{scharnhorst2022robust}. Then, we consider a controller design method that explicitly accounts for the uncertainty around the nominal model. Since the nominal model is generally nonlinear and lacks a specific structure, we consider a method in which the control law is designed so as to render the dynamics in closed loop nearly linear (as much as possible) by cancelling the nonlinearities of the system. We show that the method returns positively invariant sets for the closed-loop system and can be implemented via semidefinite programming. Kernel-based methods have been previously considered mostly in connection with Gaussian processes \cite{Umlauft2018GP,devonport2020bayesian}. In a deterministic setting, contributions have been proposed in the realm of modeling and control \cite{maddalena2021deterministic,Maddalena2021,van2022kernel,huang2022robust}. To the best of our knowledge, our work is the first work on kernel learning that gives deterministic guarantees in the context of feedback controller design.

The rest of the paper is organized as follows: Preliminaries on kernels, RKHS and regularized interpolation are given in Section \ref{s2}. Section \ref{s3} provides the main result in which we derive a controller design method based on kernel models. Section \ref{s4} presents simulation results on a nonlinear system. Conclusions and future work are discussed in Section \ref{s5}.

\textit{Notation.} Throughout the paper, $\mathbb{R}$ denotes the set of real numbers, and $\mathbb{N}_{>0}$ denotes the set of positive integers. $\mathbb S^{n \times n}$ denotes the set of real-valued symmetric matrices. Given a matrix $M$, $M \succ 0$ ($M\succeq0$) means that $M$ is positive definite (positive semidefinite), while $M \prec 0$ ($M\preceq0$) means that $M$ is negative definite (negative semidefinite). Finally, we denote by $|x|$  the 2-norm of a vector $x$, and by $\lVert M \rVert$ the induced 2-norm of a matrix $M$. 
Other, less standard, notions are introduced throughout the paper.

\section{Preliminaries}\label{s2}

\subsection{Kernels and their RKHS}
Given a non-empty set $\Omega \subseteq \mathbb{R}^n$, a continuous  function $K: \Omega \times \Omega \rightarrow \mathbb{R}$ is called a positive definite kernel on $\Omega$ if $\sum_{i, j}  \alpha_i \alpha_j K(x_i,x_j) >0$ for any $N  \in \mathbb{N}_{>0}$, any set of 
pairwise distinct points $x_1,\ldots,x_N \subseteq \Omega$, and any nonzero vector $\alpha \in \mathbb{R}^N$. 
It is called positive semidefinite if $\sum_{i, j}  \alpha_i \alpha_j K(x_i,x_j) \geq0$ for any $N  \in \mathbb{N}_{>0}$, any set of pairwise distinct points $x_1,\ldots,x_N \subseteq \Omega$, and any vector $\alpha \in \mathbb{R}^N$, 
It is called symmetric if $K(x,y)=K(y,x)$ for any $x,y\in\Omega$. 

\begin{definition}\thlabel{def1} (\cite[Def. 10.1]{wendland2004scattered})
Let $\mathcal{H}$ be a real Hilbert space of functions $f: \Omega \rightarrow \mathbb{R}$. The function $K: \Omega \times \Omega \rightarrow \mathbb{R}$ is a reproducing kernel of $\mathcal{H}$ if
\begin{enumerate}
\item For every $y \in \Omega$, the function $K(\cdot,y)$ belongs to $\mathcal{H}$.
\item (Reproducing property) For every $y \in \Omega$ and every $f \in \mathcal{H}$, it holds that
\[
f(y) = \langle f(\cdot), K(\cdot,y)\rangle_\mathcal{H} , 
\]
where $\langle \cdot, \cdot \rangle _ \mathcal{H}$ is the inner product in $\mathcal{H}$.
\end{enumerate}
\end{definition}

\begin{fact}
\thlabel{fact1}
\cite{aronszajn1950theory}
To every positive semidefinite and symmetric kernel $K$, there corresponds a unique Hilbert space admitting $K$ as a reproducing kernel. \hfill $\blacksquare$
\end{fact}  

A Hilbert space that admits a reproducing kernel is called a reproducing kernel Hilbert space (RKHS). By \thref{def1}, the kernel centred at a point $a \in \Omega$, i.e., $K(\cdot,a)$, belongs to $\mathcal{H}$. For a function of the form $f(\cdot) =\sum_{i =1}^{N} \alpha_i K(\cdot,x_i)$ where $N\in \mathbb{N}_{>0}$, $\alpha_i \in \mathbb{R}$ and $x_i \in \Omega$, we have that $f \in \mathcal{H}$ and its RKHS function norm is $\lVert f\rVert_\mathcal{H} := \sqrt{\langle f,f \rangle _ \mathcal{H}}$. Further,
{\setlength\arraycolsep{2pt} 
\begin{equation} \label{norm}
\begin{array}{rcl}
\lVert f\rVert_\mathcal{H}^2 
&=& \displaystyle \sum_{i =1}^{N} \sum_{j =1}^{N} \alpha_i \alpha_j K(x_i,x_j) .
\end{array}
\end{equation}}

\subsection{Regularized interpolation and its error bound}
Consider a positive semidefinite and symmetric reproducing kernel $K: \Omega \times \Omega \rightarrow \mathbb{R}$ and the associated RKHS $\mathcal{H}$. Consider an unknown function $f: \Omega \rightarrow \mathbb{R}$ belonging to $\mathcal{H}$, and let $f$ generate the data points $(y_i,x_i), i=0,\ldots,T-1$, where 
$y_i=f(x_i)$  Our objective is to find a function $s_f \in \mathcal{H}$ that minimizes the cost function
\begin{equation} \label{regular_cost}
\sum_{i =0}^{T-1} |y_i - s_f(x_i)|^2 + \lambda \lVert s_f \lVert _\mathcal{H}^2 ,
\end{equation}
where $\lambda > 0$ is the regularization parameter. By the representer theorem \cite{scholkopf2001generalized}, the minimizer takes the form
\begin{equation} \label{interpolant2.1}
s_f(x) =\alpha \textbf{k}(x)
\end{equation}
where $\alpha \in \mathbb{R}^T$ and 
\begin{equation}
\label{kernelbasis}
\textbf{k}(x) := \begin{bmatrix} K(x,x_0) & K(x,x_1) & \cdots & K(x,x_{T-1})  \end{bmatrix}^\top.
\end{equation}
The functions $K(x,x_i)$ are called kernel-based basis functions that are the kernels centered at the data points $x_i, i=0,\ldots,T-1$. The number of kernel-based basis functions is equal to the number of data points, and when the dataset is fixed, determining the model $s_f$ is equivalent to computing the coefficients 
$\alpha$. By \cite[Th. 2] {pillonetto2014kernel}, we have 
\begin{equation} \label{interpolant2}
s_f(x) = y_X (\lambda I_T + K_X)^{-1}\textbf{k}(x)
\end{equation}
where
\begin{equation}
y_X := \begin{bmatrix} y_0 & y_1 & \cdots & y_{T-1}  \end{bmatrix}, \\
\end{equation}
and
\begin{equation} \label{eq:KX}
\scriptsize{K_X  :=\left[ \begin{array}{cccc}
K(x_0, x_0)  & K(x_1, x_0) & \cdots & K(x_{T-1}, x_0) \\[0.1cm] 
K(x_0, x_1)  & K(x_1, x_1) & \cdots & K(x_{T-1}, x_1) \\[0.1cm]
\vdots& \vdots& \ddots& \vdots \\[0.1cm]
K(x_0, x_{T-1})  & K(x_1, x_{T-1}) & \cdots & K(x_{T-1}, x_{T-1})
\end{array}
\right]}. \normalsize
\end{equation}
The following result gives a deterministic finite-sample error bound associated with \eqref{interpolant2}.

\begin{theorem}
\thlabel{the3}
Consider a positive semidefinite symmetric reproducing kernel $K: \Omega \times \Omega \rightarrow \mathbb{R}$ with $\Omega \subseteq \mathbb{R}^n$
along with the associated RKHS $\mathcal{H}$. Let $f \in \mathcal{H}$ generate the data points $(y_i,x_i)$, $i=1,\ldots,T-1$, where $y_i = f(x_i)$. 
Then the interpolating function $s_f(x)$ in \eqref{interpolant2} provides an estimate of the function $f(x)$ for $x \in \Omega$ with interpolation error satisfying
\begin{equation}
\label{errorbound2}
\begin{array}{rl}
|f(x) - s_f(x)| \leq  \lVert f\rVert_\mathcal{H}  \sqrt{K(x,x)-\textbf{k}(x)^\top \hat K_X^{-1}\textbf{k}(x)}, \\[0.1cm]
\forall x \in \Omega. 
\end{array}
\end{equation}
where $\hat K_X := (\lambda I_T + K_X) (2\lambda I_T + K_X)^{-1} (\lambda I_T + K_X)$.
\end{theorem}

\emph{Proof.} See Appendix A. Similar statements are given  in 
\cite{fasshauer2011positive,maddalena2021deterministic}.   \hfill $\blacksquare$

\begin{remark}
The interpolation error bound \eqref{errorbound2} is factored into two parts: the first term $\lVert f\rVert_\mathcal{H}$ only depends on $f$, while the second term is independent of $f$ and only depends on the kernel and the data. Deterministic bounds of this type have been recently proposed in the literature \cite{scharnhorst2022robust}. It is not our goal here to consider the problem of deriving optimal bounds; rather, we take \eqref{errorbound2} as an example of error bounds that can be used for control design purposes. We refer the interested reader to \cite{scharnhorst2022robust} for a more ample discussion on the problem of establishing interpolation bounds.
\hfill $\blacksquare$
\end{remark}
\begin{remark}
The regularization helps to avoid overfitting. The bound in \eqref{errorbound2} continues to hold if $\lambda = 0$ \cite[Sec. 5.1] {fasshauer2011positive}, but in this case the kernel $K$ must be positive definite. \hfill 
$\blacksquare$
\end{remark}

\section{Main Results}\label{s3}

Consider a discrete-time affine-input nonlinear system
\begin{equation}
\label{sys.d}
x^+ = f(x)+Bu
\end{equation}
where $x\in \mathbb{R}^n$ is the state and $u\in \mathbb{R}^m$ is the control input, $f$ is the drift vector field,
and $B$ is a constant matrix. Both $f$ and $B$ are considered unknown. 
We instead assume that $(x_e,u_e)=(0,0)$ is a known unstable equilibrium point of the system. 
The objective is to design a feedback controller that stabilizes the dynamics around the origin.   

As anticipated in the Introduction, we will consider an \emph{indirect} method that consists of two steps: we first construct a kernel-based model of the system along with deterministic error bounds (Theorem \ref{the3}).
Then, we will derive a controller design method that explicitly accounts for the uncertainty around the nominal model. This method is inspired by \cite{NonlinearityCancellation2023} but presents some differences that will be discussed later on in the paper. 

\subsection{Kernel-based functions and error bounds}

To derive a model of the system, we proceed in two steps. As a first step, we 
set the control input $u=0$ and collect from the system a dataset
\begin{equation}\label{dataset1}
\mathbb{D} := \{x(k)\}_{k=0}^T
\end{equation}
of samples satisfying $x(k+1)=f(x(k))$, $k=0,\ldots,T-1$, with $T>0$. We note that 
the samples can be computed from a single trajectory or from multiple trajectories of the system. 
\begin{equation}
\label{eq_X0}
X_0 := \begin{bmatrix} x(0) & x(1) & \cdots & x(T-1)  \end{bmatrix}  
\end{equation}
\begin{equation} \label{eq_X1}
X_1 := \begin{bmatrix} x(1) & x(2) & \cdots & x(T) \end{bmatrix}. 
\end{equation}
Let now $K$ denote a kernel function chosen by the designer.
Given $K$ and the dataset $\mathbb{D}$, let
\begin{equation}
\textbf{k}(x) =
\begin{bmatrix} K(x,x(0)) & \hspace{-0.1cm} K(x,x(1)) & \hspace{-0.1cm}  \cdots & \hspace{-0.1cm}  K(x,x(T-1))  \end{bmatrix}^\top.
\end{equation}
The function $\textbf{k}(x)$ represents the vector of basis functions that will generate the interpolation function $s_f(x)$. 

To use Theorem \ref{the3} we need the following assumption.
\begin{assumption}\thlabel{assu1}
All the $n$ components of $f$ in \eqref{sys.d} belong to the RKHS $\mathcal{H}$ associated to $K$. 
Moreover, an upper bound 
$\Gamma_i$ for $\lVert f_i\rVert_\mathcal{H}$,
$i=1,\ldots,n$, is known.
\hfill $\blacksquare$
\end{assumption}

Methods for estimating $\Gamma_i$ are discussed in \cite{scharnhorst2022robust}. Here we just point out that the bound can be loose, although this may render the control design step more difficult. 
By solving \eqref{interpolant2}, the interpolation function of $f(x)$ takes the form
\begin{equation} \label{interpolation1}
s_f(x) = A\textbf{k}(x)
\end{equation}
where $A := X_1 (\lambda I_T + K_{X_0})^{-1}$ and where the matrix $K_{X_0}$ is as in \eqref{eq:KX} with
$X$ replaced by $X_0$. 
Let
\begin{equation} \label{error}
d(x) := f(x) - s_f(x).
\end{equation}
By \eqref{errorbound2}, each component of the vector $d$ thus satisfies
\begin{equation}
\begin{array}{rl}
|d_i(x)| \leq \lVert f_i\rVert_\mathcal{H} \sqrt{ K(x,x)-\textbf{k}(x) \hat K_{X_0}^{-1} \textbf{k}(x)},\\[0.2cm] \enspace i=1,\ldots,n, \quad \forall x \in \Omega,
\end{array}
\end{equation}
with $\hat K_{X_0}$ as in Theorem \ref{the3} with $X$ replaced by $X_0$.
Hence, by letting $\Gamma:= \begin{bmatrix} \Gamma_1 & \hspace{-0.1cm} \Gamma_2 & \hspace{-0.1cm}  \cdots & \hspace{-0.1cm}  \Gamma_n  \end{bmatrix}{}  ^\top$ and defining
\begin{equation}
\label{errorbound_vec2}
\delta(x):= |\Gamma| \sqrt{ K(x,x)-\textbf{k}(x) \hat K_{X_0}^{-1}\textbf{k}(x)}, 
\end{equation}
if follows from \thref{assu1} that the interpolation error on the function $f$ satisfies the deterministic bound
\begin{equation} \label{bound1}
|d(x)| \leq \delta(x), \quad \forall x \in \Omega. 
\end{equation}

\subsection{Controller design method based on approximate nonlinearity cancellation}

As a second step, we derive a control design method that exploits the bound 
on the interpolation error. By previous analysis, the dynamics \eqref{sys.d} can be written as 
\begin{equation} \label{new.sys.d}
x^+ = A\textbf{k}(x) + Bu + d(x)
\end{equation}
where $A \in \mathbb{R}^{n\times T}$ is known and $B\in \mathbb{R}^{n\times m}$ is still unknown. 

To determine the feedback controller, we make a second experiment on the system where we apply a nonzero input sequence $u$ and collect a new dataset 
\begin{equation}
\overline {\mathbb{D}} := \{\overline{x}(k), u(k)\}_{k=0}^{\overline{T}}
\end{equation}
of samples satisfying $\overline{x}(k+1)=f(\overline{x}(k))+Bu(k)$, where $k=1,\cdots,\overline{T}$ 
and $\overline{T}>0$.
These data are grouped in the data matrices
\begin{subequations}
\begin{align}
&\overline X_0 := \begin{bmatrix} \overline{x}(0) & \overline{x}(1) & \cdots & \overline{x}(\overline T-1)  \end{bmatrix} \in \mathbb{R}^{n\times \overline T} 
\\
&\overline X_1 := \begin{bmatrix} \overline{x}(1) & \overline{x}(2) & \cdots & \overline{x}(\overline T)  \end{bmatrix}  \in \mathbb{R}^{n\times \overline T}
\\
&U_0 := \begin{bmatrix} u(0) & u(1) & \cdots & u(\overline T-1)  \end{bmatrix} \in \mathbb{R}^{m\times \overline T} 
\label{eq:U0} \\
&K_0 := \begin{bmatrix} \textbf{k}(x(0)) & \textbf{k}(x(2)) & \cdots & \textbf{k}(x(\overline T-1))  \end{bmatrix} \in \mathbb{R}^{T\times \overline T} \label{data.4}
\end{align}
\end{subequations}
which satisfy the identity 
\begin{equation} 
\label{data.identity}
\overline X_1 = A K_0 + BU_0 + D_0
\end{equation}
where 
\[
D_0 := \begin{bmatrix} d(0) & d(1) & \cdots & d(\overline T-1) \end{bmatrix} 
\]
is the (unknown) data matrix of samples of $d$. 

We assume that this second experiment is carried out with an input such that the corresponding matrix $U_0$ has full row rank. This can be interpreted as an excitation condition on the experiment. We will write this condition as an assumption but it is indeed a \emph{design} condition.

\begin{assumption}\thlabel{assu3}
$U_0$ has full row rank. \hfill $\blacksquare$
\end{assumption}

By letting $\hat X_1 := \overline X_1 - A K_0$, we have $BU_0 = \hat X_1 -D_0$.
\thref{assu3} thus implies
\begin{equation}
\label{B}
B = (\hat X_1 -D_0) \underbrace{U_0^\top (U_0 U_0^\top)^{-1}}_{=: U_0^\dag}
\end{equation}
and the dynamics can be written as
\begin{equation} \label{newnew.sys.d}
x^+ = A\textbf{k}(x) + (\hat X_1 -D_0)U_0^\dag u + d(x).
\end{equation}
Arrived at this point, note that the dynamics of $\textbf{k}(x)$ depend on the selected kernel. We will consider the general case in which $\textbf{k}(x)$ consists of both linear and nonlinear functions, so that $A\textbf{k}(x)$ can be decomposed as $A\textbf{k}(x) = \overline A x + \hat A \hat {\textbf{k}} (x)$ with $\hat {\textbf{k}}: \mathbb{R}^n \rightarrow \mathbb{R}^S$ that contains only nonlinear functions. The special case $\textbf{k}(x) = x$, gives $\hat A = 0_{n\times S}$. In contrast, $\overline A = 0_{n\times n}$ when $\textbf{k}(x)$ contains only nonlinear functions. Note that for a fixed $\textbf{k}(x)$, the choice of $\hat {\textbf{k}} (x)$ is not unique, and different choices of $\hat {\textbf{k}} (x)$ generate different matrices $\hat A$. 
With this decomposition,  \eqref{newnew.sys.d} reads equivalently as
\begin{equation} \label{4}
x^+ = \overline A x + \hat A \hat {\textbf{k}}(x) + (\hat X_1 -D_0)U_0^\dag u + d(x).
\end{equation}
This decomposition suggests a control law in the form
\begin{equation} \label{control}
u = \overline K x + \hat K \hat {\textbf{k}}(x)
\end{equation}
which gives the closed-loop dynamics
{\setlength\arraycolsep{2pt}  
\begin{eqnarray} \label{closeloop}
x^+ &=& (\overline A +( \hat X_1 -D_0)U_0^\dag \overline K)x \nonumber \\  
&& + (\hat A + (\hat X_1 -D_0)U_0^\dag \hat K)\hat {\textbf{k}}(x) + d(x).
\end{eqnarray}}%

A natural way to design the control law is then to design $\overline K$ so as to stabilize the linear part of the dynamics, and to design $\hat K$ so as to try to cancel out the nonlinear terms. This approach has been originally proposed in \cite{NonlinearityCancellation2023}, and we refer the reader to it for a discussion regarding the connections between this approach and the classic feedback linearization.
By Lyapunov theory, a necessary and sufficient condition for 
the linear dynamics $\dot{\xi}=(\overline A +( \hat X_1 -D_0)U_0^\dag \overline K)\xi$ to be stable is that for any $Q\succ0$ there exists a matrix $S\succ0$ that solves the Lyapunov equation
\begin{equation}
\begin{array}{rl}
(\overline A +( \hat X_1 -D_0)U_0^\dag \overline K)^\top S (\overline A +( \hat X_1 -D_0)U_0^\dag \overline K) \\[0.1cm] - S + S Q S \preceq 0.
\end{array}
\end{equation}
Letting $P=S^{-1}$ and multiplying both sides by $P$, this turns out to be equivalent to 
\begin{equation} \label{eq:stab_lin}
\begin{array}{rl}
(\overline A P +( \hat X_1 -D_0)U_0^\dag Y)^\top P^{-1} (\overline A P +( \hat X_1 -D_0)U_0^\dag Y) \\[0.1cm] - P + Q \preceq 0
\end{array}
\end{equation}
having set $Y=\overline{K}P$. As we will see, this form is particularly convenient because 
it can be expressed as a linear matrix inequality (LMI) constraint. However, we cannot implement directly \eqref{eq:stab_lin} because $D_0$ is unknown. The idea is thus to ensure that the constraint is satisfied for all the matrices $D$ in a given set $\mathcal{D}$ to which $D_0$ is known to belong, i.e., 
\begin{equation}\label{closeloop4}
\begin{array}{rl}
(\overline AP +( \hat X_1 -D)U_0^\dag Y)^\top P^{-1} (\overline AP +( \hat X_1 -D)U_0^\dag Y)\\[0.1cm] - P + Q  \preceq 0 \quad \forall D \in \mathcal{D}.
\end{array}
\end{equation}
Let
\begin{equation}
\label{Delta2}
\Delta := \left( \sum_{k =0}^{\overline T-1}\delta(\overline{x}(k))^2 I_n \right)^{1/2}.
\end{equation}
Since 
$D_0D_0^\top \preceq \Delta^2$,
we can therefore solve \eqref{closeloop4} with respect to the set 
\begin{equation}
\mathcal{D} := \{D \in \mathbb{R}^{n\times \overline T}: D D^\top \preceq \Delta^2 \}.
\end{equation}

Condition \eqref{closeloop4} cannot be implemented directly because it involves infinitely many constraints. The next result provides a tractable (and convex) condition for \eqref{closeloop4}.
\begin{lemma}
\thlabel{lemma1}
Given $Q \succ 0$ and $\Delta$ defined in \eqref{Delta2}, if there exist $P\in \mathbb S^{n \times n}$, $Y\in \mathbb R^{m \times n}$ and a scalar $\epsilon>0$ such that
\begin{equation}
\label{closeloop5}
\left[ \begin{array}{ccc}
 P-Q  & (  \overline A P+\hat X_1 U_0^\dag Y)^\top &  (U_0^\dag Y)^\top \\[0.1cm] 
\overline A P+\hat X_1 U_0^\dag Y &   P-\epsilon \Delta^2 & 0_{n \times \overline T}
\\[0.1cm]
U_0^\dag Y & 0_{ \overline T \times n} & \epsilon I_{\overline T}
\end{array}
\right] \succeq 0
\end{equation}
then \eqref{closeloop4} holds.
\end{lemma}
\begin{proof}
See Appendix B.  
\end{proof}

Condition \eqref{closeloop5} guarantees stability of the linear dynamics $\dot{\xi}=(\overline A +( \hat X_1 -D_0)U_0^\dag \overline K)\xi$ with $\overline{K}=YP^{-1}$. The remaining part of the controller, i.e., the matrix $\hat{K}$, can be determined so as to minimize the effect of the nonlinearities in the closed loop. Including the design of $\overline K$, a prototypical formulation is the following: 
\begin{subequations}
\label{eq:SDP}
\begin{align}
\textrm{minimize}_{ P, Y, \hat K, \epsilon} \quad
& \| \hat A + \hat X_1 U_0^\dag \hat K \| + \alpha\,\| P \| 
\label{eq:SDPa1} \\ 
\textrm{subject to} \quad
&
\eqref{closeloop5}
\label{eq:SDP1} 
\end{align}
\end{subequations}
where $\alpha\geq0$ is a design parameter. As shown, \eqref{closeloop5} ensures
stability of the linear dynamics $\dot{\xi}=(\overline A +( \hat X_1 -D_0)U_0^\dag \overline K)\xi$. Instead, 
minimizing $\| \hat A + \hat X_1 U_0^\dag \hat K \|$ tries to reduce as much as possible the effect of the nonlinearities in the closed loop. In this context, the term $\alpha\,\| P \| $ acts as a regularization term that permits to enlarge the estimate of the positive invariant set for the closed-loop dynamics, as detailed in the sequel. Before proceeding, we remark that \eqref{eq:SDP} should be viewed as an example. An alternative is to explicitly account for $D_0$ for the nonlinear term as well:
\begin{subequations}
\label{eq:SDP_variant}
\begin{align}
\textrm{minimize}_{ P, Y, \hat K, \epsilon} \quad
& \| \hat A + (\hat X_1-D) U_0^\dag \hat K \| + \alpha\,\| P \| 
\label{eq:SDPa2} \\ 
\textrm{subject to} \quad
& \eqref{closeloop5},\, D \in \mathcal{D}.
\label{eq:SDP2} 
\end{align}
\end{subequations}
Also this problem can be cast as a semidefinite program.

The rest of this section is devoted to show that this method guarantees the existence of a positively invariant set for the closed loop if the modelling error is sufficiently small.
\begin{definition}\thlabel{def3}
For the system $x^+ = f(x)$, if for every $x(0) \in \mathcal{S}$, it holds that $x(t) \in \mathcal{S}$ for $t > 0$, then $\mathcal{S}$ is called a positively invariant (PI) set. \hfill $\blacksquare$
\end{definition}
 
Let $V(x) = x^\top P ^{-1} x$, which acts as a Lyapunov function for the linear part of the dynamics, and define for brevity $\Psi = \overline A +( \hat X_1 -D_0)U_0^\dag \overline K$ and $\Xi = \hat A + (\hat X_1 -D_0)U_0^\dag \hat K$. Then, the Lyapunov function satisfies
\begin{align*}
 &V(x^+) - V(x) \\&=(\Psi x + \Xi \hat {\textbf{k}}(x) + d(x))^\top P ^{-1} (\Psi x + \Xi \hat {\textbf{k}}(x) + d(x))  \\ & \hspace*{15pt}- x^\top P ^{-1} x
\end{align*}
Bearing in mind the expressions of $\Psi$ and $\Xi$, the fact that
$D_0D_0^\top \preceq \Delta^2$, and $|d(x)| \leq \delta(x)$, simple (although tedious) calculations give
\begin{equation} \label{RPI3}
V(x^+) - V(x) \leq l(x) + g(x,\delta(x))
\end{equation}
where
\begin{subequations}
\begin{align*}
&l(x) := -x^\top P^{-1} Q P^{-1} x + l_1(x) + l_2(x) + l_3(x) + l_4(x)
\\
&l_1(x) := (2(\overline A + \hat X_1 U_0^\dag \overline K)x \\& \hspace*{35pt} + (\hat A + \hat X_1 U_0^\dag \hat K) \hat {\textbf{k}}(x))^\top P ^{-1} (\hat A + \hat X_1 U_0^\dag \hat K) \hat {\textbf{k}}(x)
\\
&l_2(x) := \| \Delta \| |(2(\overline A + \hat X_1 U_0^\dag \overline K)x \\& \hspace*{35pt} + (\hat A + \hat X_1 U_0^\dag \hat K) \hat {\textbf{k}}(x))^\top P^    {-1}| |U_0^\dag \hat K \hat {\textbf{k}}(x)|
\\
&l_3(x) := \| \Delta \| |2 U_0^\dag \overline K x +  U_0^\dag \hat K  \hat {\textbf{k}}(x)| | P^    {-1} (\hat A + \hat X_1 U_0^\dag \hat K) \hat {\textbf{k}}(x)|
\\
&l_4(x) := \| \Delta \| ^2 \| P^{-1} \| |2 U_0^\dag \overline K x +  U_0^\dag \hat K  \hat {\textbf{k}}(x)| |U_0^\dag \hat K \hat {\textbf{k}}(x)|
\\
&g(x,\delta(x)) := r_1(x) \delta(x) + r_2(x) \delta(x) +r_3 \delta(x)^2
\\
&r_1(x) := 2|((\overline A + \hat X_1 U_0^\dag \overline K)x + (\hat A + \hat X_1 U_0^\dag \hat K) \hat {\textbf{k}}(x))^\top P^    {-1}|
\\
&r_2(x) := 2\| \Delta \|  \| P^{-1} \| | U_0^\dag \overline K x +  U_0^\dag \hat K  \hat {\textbf{k}}(x)|
\\
&r_3 := \| P^{-1} \|.
\end{align*}
\end{subequations}
(These expressions show that penalizing the term $\|P\|$ in \eqref{eq:SDP} may increase the estimate of the PI set 
since $\ell(x)$ scales with $P^{-2}$ 
while
the other terms scale with $P^{-1}$).  

Let 
\begin{equation} \label{X}
\mathcal{X} := \{x: l(x) + g(x,\delta(x)) \leq 0\}
\end{equation}
and let $\mathcal{X}^c$ be its complement. Let $\mathcal{R}_\gamma := \{x: V(x) \leq \gamma\}$, where $\gamma >0$ is arbitrary, and define $\mathcal{Z} := \mathcal{R}_\gamma \cap \mathcal{X}^c$, which characterizes all the points in $\mathcal{R}_\gamma$ for which the Lyapunov difference $V(x^+) - V(x)$ can be positive. Then the following main result holds.
\begin{theorem}
\thlabel{the2} 
Consider a nonlinear system as in \eqref{new.sys.d}, and suppose that \eqref{eq:SDP} is feasible with a given $Q \succ 0$ and where $\Delta$ is defined in \eqref{Delta2}. Consider the closed-loop system with the controller \eqref{control} obtained from \eqref{eq:SDP}. If 
\begin{equation} \label{RPI4}
V(x) + l(x) + g(x,\delta(x)) \leq \gamma \quad \forall x \in \mathcal{Z}
\end{equation}
then $\mathcal{R}_\gamma$ is a PI set for the closed-loop system.

\begin{proof}
Suppose \eqref{RPI4} holds and let $x \in \mathcal{R}_\gamma$. The analysis can be divided in two sub-cases. If $x \notin \mathcal{Z}$ then $x \in \mathcal{X}$ and $V(x^+) - V(x) \leq l(x) + g(x,\delta(x)) \leq 0$. Hence, $x^+ \in \mathcal{R}_\gamma$. If $x \in \mathcal{Z}$ then $V(x^+) \leq V(x) + l(x) + g(x,\delta(x)) \leq \gamma$ in view of \eqref{RPI3} and \eqref{RPI4}. Hence, $x^+ \in \mathcal{R}_\gamma$.
\end{proof}
\end{theorem}

We close this section with a few remarks. The first remark regards the comparison with \cite{NonlinearityCancellation2023}. In  \cite[Th. 8]{NonlinearityCancellation2023}, a similar result is given that takes unmodeled dynamics into account. In this respect, the results presented here give a systematic principled method for bounding modelling errors. \cite[Th. 7]{NonlinearityCancellation2023} also shows that asymptotic 
stability follows when the error bound $\delta(x)$ satisfies $\lim_{|x| \rightarrow \infty} \frac{\delta(x)}{|x|} = 0$, e.g. when $\delta(x)$ acts as remainder in a power series expansion of $f$ about $0$. The same result holds also here but we have to bear in mind that the condition $\lim_{|x| \rightarrow \infty} \frac{\delta(x)}{|x|} = 0$ may fail to hold depending on the choice of the kernel function. In any case, invariance sets provide a safe region where we can perform additional experiments to estimate regions of attraction. 

The second remark concerns the experimental conditions. Here we have assumed noise-free data, but bounds similar to the one in \eqref{errorbound2} can be given also in case of noisy data \cite{scharnhorst2022robust}. Such bounds can be combined with existing tools for robust controller design (\emph{cf.} \cite[Sec. VI]{NonlinearityCancellation2023}) to extend the results presented in this paper.

\section{Numerical example}\label{s4}

Consider the nonlinear system 
\begin{subequations}\label{ex1}
\begin{equation}\label{ex1.1}
x_1^+ = x_2 + x_1^3 + u
\end{equation}    
\begin{equation}
x_2^+ = 0.5x_1 + 0.2x_2^2.
\end{equation}
\end{subequations}
We consider a polynomial kernel of the degree 3:
\begin{equation} \label{kernel}
K(x,y) := x^\top y + (x^\top y)^2 +  (x^\top y)^3.
\end{equation}
We set $u = 0$ and collect a dataset $\mathbb{D}$ containing $T = 10$ samples by performing multiple one-step experiments with initial states uniformly distributed in $[ -2,2 ]$. The resulting matrix $X_0$ is shown in \eqref{ex.data1}.
\begin{figure*}
\small{\begin{subequations}
\label{ex.data}
\begin{align}
&X_0 =
\begin{bmatrix}
   -0.3319  & -1.9995 &  -1.4130 &  -1.2550  & -0.4129  & -0.3232 &  -1.1822  & -1.8904 &  -0.3308   &-1.4385 \\
    0.8813 &  -0.7907  & -1.6306 &  -0.6178 &   0.1553 &   0.7409 &   1.5125 &   0.6819  &  0.2348  & -1.2076
\end{bmatrix} \label{ex.data1}\\
&\overline X_0 =
\begin{bmatrix}
   -1.5907  &  0.7776 &  -1.8002 &   0.6552&    1.7784 &   1.6136 &  -1.4429  & -0.4093  &  1.7100   & 1.0032 \\
   -0.3438  & -0.3433 &   0.1436  &  0.0596  &  0.3462  & -1.4501 &   1.2296  & -1.3386  & -0.6089  &  0.9040
\end{bmatrix}\\
&\overline X_1 =
\begin{bmatrix}
   -4.8491 &   0.3057 &  -5.9786 &   0.1063   & 5.9622  &  2.3047 &  -1.7003   &-1.7604 &   4.4809  &  2.1135 \\
   -0.7717 &   0.4124  & -0.8960 &   0.3283  &  0.9132  &  1.2274  & -0.4191 &   0.1537  &  0.9292  &  0.6651
\end{bmatrix}\\
&U_0 =
\begin{bmatrix}
-0.4806  &  0.1788 &  -0.2884 &&  -0.2345 &  -0.0084 &  -0.4466  &  0.0741 &  -0.3533 &   0.0893  &  0.1998
\end{bmatrix} 
\end{align}
\hrulefill
\end{subequations}}
\end{figure*}
With these data we construct the vector $\textbf{k}(x)$ of basis functions. 
The kernel $K(x,y)$ is symmetric positive semidefinite and there exists a unique RKHS $\mathcal{H}$ that admits  $K(x,y)$ as a reproducing kernel by \thref{fact1}. We just need to show that the nonlinear dynamics $f_1(x) = x_2 + x_1^3$ and $f_2(x) = 0.5x_1 + 0.2x_2^2$ in \eqref{ex1} are members of $\mathcal{H}$. By \thref{def1}, all of the components of $\textbf{k}(x)$ belong to $\mathcal{H}$. Then, it is sufficient to show that $f_1(x)$ and $f_2(x)$ are linear combinations of $\textbf{k}(x)$. Denote by $M(x)$ the vector of all monomials up to degree 3. We can write $f_1(x) = c_1 M(x)$, $f_2(x) = c_2 M(x)$ and $\textbf{k}(x) = M_\textbf{k} M(x)$. Note that when the matrix $M_\textbf{k}$ has full column rank, there exists $\alpha_i$ such that $c_i = \alpha_i M_\textbf{k}, \enspace i=1, 2$, and this implies that $f_1(x)$ and $f_2(x)$ can be written as the linear combinations of $\textbf{k}(x)$. Hence, the collected data in $\mathbb{D}$ should satisfy the condition that the corresponding matrix $M_\textbf{k}$ is full column rank, and this condition is indeed satisfied for the collected samples in \eqref{ex.data1}. 
Finally, in order to find an upper bound $\Gamma$ on $\lVert f\rVert_\mathcal{H}$ as in Assumption 1, we compute $\lVert f\rVert_\mathcal{H}$ explicitly.  By \eqref{norm}, we have $\lVert f_1\rVert_\mathcal{H} = \alpha_1 K_{X_0} \alpha_1^\top = 2$ and $\lVert f_2\rVert_\mathcal{H} = \alpha_2 K_{X_0} \alpha_2^\top = 0.29$. For controller design we select $\Gamma_1 = 3$ and $\Gamma_2 = 0.4$, which over-approximate the true values by more than $30\%$.
Finally, we select $\lambda = 10^{-7}$. We note that large values of $\lambda$ results in large bounds $\delta(x)$ 
(Theorem \ref{the3}), and this may eventually render the controller design program infeasible.

Next, we collect a dataset $\overline {\mathbb{D}}$ containing $\overline T = 10$ samples by performing again multiple one-step experiments with input uniformly distributed in  $[-0.5,0.5]$,
and with initial states within $[-2,2]$. The resulting data matrices $ \overline X_0$, $ \overline X_1$ and $U_0$ are reported in \eqref{ex.data}, from which we compute the two matrices $K_0$ and $U_0^\dag$ as in  \eqref{data.4} and \eqref{B}, respectively. Note that the first term of $K(x,y)$, i.e. $x^\top y$, produces the linear part of $A\textbf{k}(x)$, and gives
\[\begin{array}{rl}
\overline A x=& A\begin{bmatrix} x^\top x(0) & x^\top x(1) & \cdots & x^\top x(T-1)  \end{bmatrix}^\top
\\[0.1cm]
=&
A\begin{bmatrix} x(0) & x(1) & \cdots & x(T-1)  \end{bmatrix}^\top x
\\[0.1cm]
=&
A X_0^\top x,
\end{array}
\]
and thus $\overline A = A X_0^\top$. In addition, we set 
\[\begin{array}{rl}
\hat {\textbf{k}}(x) := \lbrack (x^\top x(0))^2 +  (x^\top x(0))^3 \,\, (x^\top x(1))^2 +  (x^\top x(1))^3 \\[0.1cm]
\cdots \,\, (x^\top x(T-1))^2 +  (x^\top x(T-1))^3  \rbrack  ^\top
\end{array}\]
and thus $\hat A = A$.
We solve \eqref{eq:SDP} with $Q = I_2$, and $\alpha=1$. The resulting controller along with the matrix $\hat A + \hat X_1 U_0^\dag \hat K$ are reported in \eqref{ex.controller} on the next page. For the  dynamics not 
depending
on $D_0$ in \eqref{closeloop}, we obtain
\begin{equation}
\label{deterministic}
(\overline A + \hat X_1 U_0^\dag \overline K)x + (\hat A + \hat X_1 U_0^\dag \hat K)\hat {\textbf{k}}(x) = \begin{bmatrix}
0.2481 x_2\\
0.5 x_1 + 0.2 x_2^2
\end{bmatrix}
\end{equation}
We note that the program \eqref{eq:SDP} correctly forces $u$ to cancel out the nonlinearity in \eqref{ex1.1}.

For this controller, we numerically
determine the set $\mathcal{X} = \{ \xi: l(\xi) + g(\xi,\delta(\xi)) \leq 0\}$.
Any sub-level set $\mathcal{R}_\gamma$ of the Lyapunov function $V(x) = x^\top P ^{-1} x$ contained in $\mathcal{X} \cup \{0\}$ and satisfying \eqref{RPI4} gives an estimate of the PI set for the closed-loop system. The set $\mathcal{X}$ and a sublevel set of $V$ are shown in Figure \ref{RPI}. We can numerically verify that the PI set in Figure \ref{RPI} is also a region of attraction (ROA), and one possible reason is that both $\hat {\textbf{k}}(x)$ and $\delta(x)$ converge to 0 when $x$ converges to 0 since we use the polynomial kernel $K(x,y)$. Remarkably, the obtained estimate of the ROA is almost the same as the one obtained in \cite{NonlinearityCancellation2023} with knowledge of the true basis functions. 

 \begin{figure}[ht!]
\includegraphics[scale=0.55]{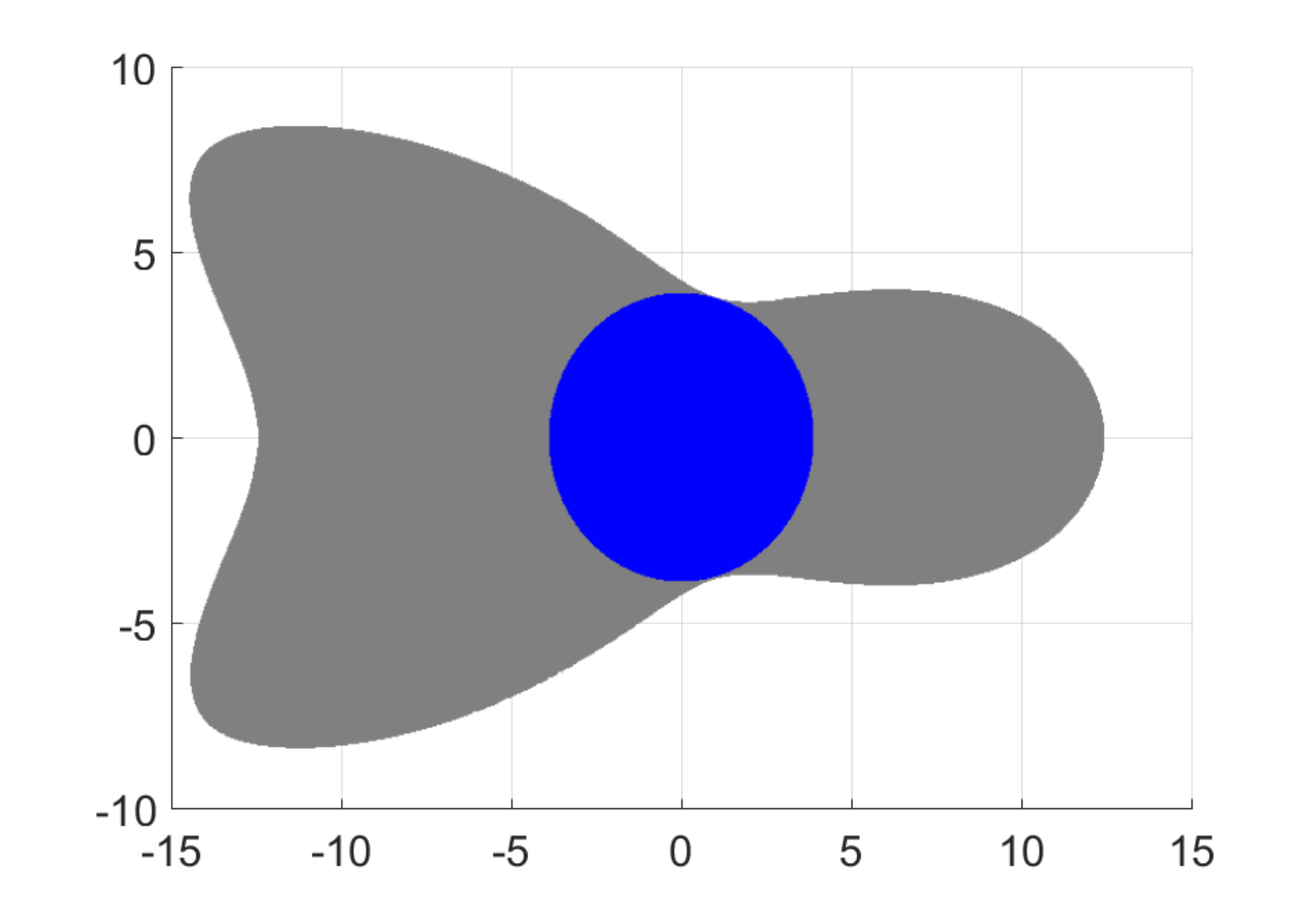}
   \caption{The grey
set represents the set $\mathcal{X}$, while the blue set is the PI set $\mathcal{R}_\gamma$; here, $P = \begin{bmatrix}
1.3350 & 0\\
0 & 1.3350
\end{bmatrix}$ and $\gamma = 11.5$. We observe that the set $\mathcal{Z}$ is empty and hence the PI set also provides an estimate of the ROA for the closed-loop system.}\label{RPI}
\end{figure}

\begin{figure*}
\small{\begin{subequations}
\label{ex.controller}
\begin{align}
&\overline K = \begin{bmatrix} 0 & -0.7519 \end{bmatrix}
\\ &\hat K = \begin{bmatrix} 15.0008  &  0.3255 &   0.2983 &  -1.9372 &  11.0185 & -20.2206 &  -0.3585 &   0.2566 &  -4.7991  & -0.0753 \end{bmatrix}\\
&\hat A + \hat X_1 U_0^\dag \hat K =
\begin{bmatrix}
   -0.0001 &   0  &  0  &  0.0001  & -0.0001  &  0.0002 &   0  & 0 &  0 &  -0.0001 \\
   -2.8205  &  0.2012 &   0.7892  &  1.8635  & -3.3635 &   4.2092 &   0.1724 &  -0.2044 &  -0.6411   &-2.0833
\end{bmatrix}
\end{align}
\hrulefill
\end{subequations}}
\end{figure*}

\section{Conclusions}\label{s5}

We have investigated the problem of designing feedback controllers for affine-input nonlinear systems from data using kernel learning techniques. We have considered a method in which a nominal model of the system is determined using kernel-based functions, along with an explicit upper bound on the modelling error. Then, a controller design method is proposed that involves the solution of a semidefinite program. We have shown that the method ensures, despite the presence of unmodeled dynamics, the existence of positively invariant sets for the closed-loop dynamics. An important venue for future research is the problem of understanding what kernels are more suited for control goals.

\section{Appendix}

\subsection{Proof of Theorem \ref{the3}}

Let $u(x) := (\lambda I_T + K_X)^{-1}\textbf{k}(x)$ and one obtains $s_f(x) = \sum_{i =0}^{T-1} y_i u_i(x)$ by \eqref{interpolant2}. By the reproducing property of $f(x)$ and by recalling that $y_i=f(x_i)$, the modelling error satisfies
\[
\begin{array}{ll}
& |f(x) - s_f(x)| \\[0.1cm]
& \quad =
|\langle f(\cdot), K(\cdot,x)\rangle _ \mathcal{H} - \sum_{i =0}^{T-1} \langle f(\cdot), K(\cdot,x_i)\rangle _ \mathcal{H} u_i(x)| \\[0.1cm]
& \quad =
|\langle f(\cdot), K(\cdot,x) - \sum_{i =0}^{T-1} K(\cdot,x_i) u_i(x) \rangle _ \mathcal{H}|.
\end{array}
\]
By the Cauchy-Schwartz inequality, 
\[
\begin{array}{rl}
& |\langle f(\cdot), K(\cdot,x) - \sum_{i =0}^{T-1} K(\cdot,x_i) u_i(x) \rangle _ \mathcal{H}|  \\[0.1cm]
& \quad \leq \lVert f\rVert_\mathcal{H} \lVert K(\cdot,x) - \sum_{i =0}^{T-1} K(\cdot,x_i) u_i(x) \rVert_\mathcal{H},  
\end{array}
\]
with the second term satisfying
\begin{equation}
\label{prof1}
\begin{array}{rl}
&\lVert K(\cdot,x) - \sum_{i =0}^{T-1} K(\cdot,x_i) u_i(x) \rVert_\mathcal{H}^2
\\[0.1cm]
& \quad = 
\langle  K(\cdot,x) - \sum_{i =0}^{T-1} K(\cdot,x_i) u_i(x), 
\\[0.1cm]
& \qquad K(\cdot,x)  - \sum_{i =0}^{T-1} K(\cdot,x_i) u_i(x) \rangle _\mathcal{H}
\\[0.1cm]
& \quad = 
\lVert K(\cdot,x) \rVert_\mathcal{H}^2 -  2 \sum_{i =0}^{T-1} u_i(x) \langle K(\cdot,x), K(\cdot,x_i)   \rangle_\mathcal{H} 
\\[0.1cm] 
& \qquad + \sum_{i =0}^{T-1} \sum_{j =0}^{T-1} u_i(x) u_j(x) \langle K(\cdot,x_i), K(\cdot,x_j)   \rangle_\mathcal{H}.
\end{array}
\end{equation}
Since
$\langle K(\cdot,x), K(\cdot,y) \rangle_\mathcal{H} = K(x,y)$ for any $x,y \in \Omega$,  we further obtain
\[
\begin{array}{rl}
&\sum_{i =0}^{T-1} u_i(x) \langle K(\cdot,x), K(\cdot,x_i)   \rangle_\mathcal{H} \\[0.1cm] & \quad =  \textbf{k}(x)^\top u(x) = \textbf{k}(x)^\top (\lambda I_T + K_X)^{-1} \textbf{k}(x)   
\end{array}
\]
and
\[
\begin{array}{rl}
&\sum_{i =0}^{T-1} \sum_{j =0}^{T-1} u_i(x) u_j(x) \langle K(\cdot,x_i), K(\cdot,x_j)   \rangle_\mathcal{H} \\[0.1cm] & \quad = u(x)^\top K_X^{-1} u(x)  \\[0.1cm] & \quad = \textbf{k}(x)^\top (\lambda I_T + K_X)^{-1} K_X (\lambda I_T + K_X)^{-1} \textbf{k}(x).
\end{array}
\]
Then, \eqref{prof1} becomes
\[
\begin{array}{rl}
&\lVert K(\cdot,x) - \sum_{i =0}^{T-1} K(\cdot,x_i) u_i(x) \rVert_\mathcal{H}^2 
\\[0.1cm] & \quad = K(x,x)-\textbf{k}(x) ^\top (2(\lambda I_T + K_X)^{-1} \\[0.1cm] & \quad - (\lambda I_T + K_X)^{-1} K_X (\lambda I_T + K_X)^{-1})\textbf{k}(x).
\end{array}
\]
Finally, note that 
\[
\begin{array}{rl}
&2(\lambda I_T + K_X)^{-1} - (\lambda I_T + K_X)^{-1} K_X (\lambda I_T + K_X)^{-1}
\\[0.1cm] & =
(\lambda I_T + K_X)^{-1} (2 \lambda I_T + K_X) (\lambda I_T + K_X)^{-1}
= \hat K_X^{-1}
\end{array}
\]
which implies
\[
\lVert K(\cdot,x) - \textbf{k}(\cdot)^\top u(x) \rVert_\mathcal{H}^2 = K(x,x)-\textbf{k}(x) ^\top \hat K_X^{-1}\textbf{k}(x).    
\]
This gives the desired result. \hfill $\blacksquare$

\subsection{Proof of Lemma 1} 
To prove Lemma \ref{lemma1}, we need the following result.
\begin{lemma}
\thlabel{lemma2} 
\cite[Lm. 4]{NonlinearityCancellation2023}
Let $E\in \mathbb{R}^{n\times p}$, $F\in \mathbb{R}^{q\times n}$, $\Delta \in \mathbb{R}^{p\times s}$ be given, and let $\mathcal{D} := \{D \in \mathbb{R}^{p\times q}: D D^\top \preceq \Delta \Delta^\top\}$. Then, for any $\epsilon > 0$, it holds that
\[
ED^\top F + F^\top D E^\top \preceq  \epsilon^{-1} EE^\top + \epsilon F^\top \Delta \Delta^\top F, \,\, \forall D \in \mathcal{D}.
\]
\hfill $\blacksquare$
\end{lemma}

We now proceed with the proof of Lemma \ref{lemma1}.
Let \eqref{closeloop5} hold. By applying a Schur complement to \eqref{closeloop5}, we obtain 
\[
\begin{array}{rr}
&\left[ \begin{array}{cc}
 P-Q  & (  \overline A P+\hat X_1 U_0^\dag Y)^\top \\[0.1cm] 
\overline A P+\hat X_1 U_0^\dag Y &  P
\end{array}
\right] \\[0.2cm] &- \epsilon^{-1} \underbrace{\begin{bmatrix} (U_0^\dag Y)^\top \\ 0_{n \times \overline T} 
\end{bmatrix}}_{=: E}  \begin{bmatrix} U_0^\dag Y & 0_{\overline T \times n} \end{bmatrix}  
 \\[0.2cm] & -\epsilon \underbrace{\begin{bmatrix} 0_{n \times n} \\ I_n \end{bmatrix}}_{=: F^\top}  \Delta^2  \begin{bmatrix} 0_{n \times n} & I_n \end{bmatrix} \succeq 0.  
\end{array}
\]
By \thref{lemma2}, we further have
\[
\begin{array}{rl}
&\left[ \begin{array}{cc}
 P-Q  & (  \overline A P+\hat X_1 U_0^\dag Y)^\top \\[0.1cm] 
\overline A P+\hat X_1 U_0^\dag Y &  P
\end{array}
\right] \\[0.2cm] & \quad - \begin{bmatrix} (U_0^\dag Y)^\top \\ 0_{n \times \overline T} 
\end{bmatrix} D^\top \begin{bmatrix} 0_{n \times n} & I_n  \end{bmatrix}   
 \\[0.2cm] & \quad - \begin{bmatrix} 0_{n \times n} \\ I_n \end{bmatrix} D \begin{bmatrix} U_0^\dag Y & 0_{\overline T \times n} \end{bmatrix} \succeq 0, \quad \forall D \in \mathcal{D}.
\end{array}
\]
By applying a Schur complement to this LMI, we get \eqref{closeloop4}. 
This gives the desired result. \hfill $\blacksquare$

\bibliographystyle{unsrt}
\bibliography{conference_zhongjie_27-03-23}

\begin{thebibliography}{10}

\bibitem{pid}
J.~Ziegler and N.~Nichols.
\newblock Optimum settings for automatic controllers.
\newblock {\em Transactions of the American Society of Mechanical Engineers},
  64:759--768, 1942.

\bibitem{dean2019sample}
S.~Dean, H.~Mania, N.~Matni, B.~Recht, and S.~Tu.
\newblock On the sample complexity of the linear quadratic regulator.
\newblock {\em Foundations of Computational Mathematics}, 20:633--679, 2020.

\bibitem{umenberger2019robustLCSS}
M.~Ferizbegovic, J.~Umenberger, H.~Hjalmarsson, and T.~Sch{\"o}n.
\newblock Learning robust {LQ}-controllers using application oriented
  exploration.
\newblock {\em IEEE Control Systems Letters}, 4(1):19--24, 2020.

\bibitem{de2019formulas}
C.~De Persis and P.~Tesi.
\newblock Formulas for data-driven control: Stabilization, optimality, and
  robustness.
\newblock {\em IEEE Transactions on Automatic Control}, 65(3):909--924, 2019.

\bibitem{berberich2020robust}
J.~Berberich, A.~Koch, C.~Scherer, and F.~Allg{\"o}wer.
\newblock Robust data-driven state-feedback design.
\newblock In {\em 2020 American Control Conference (ACC)}, pages 1532--1538.
  IEEE, 2020.

\bibitem{henk-ddctr-uncer}
H.~{van Waarde}, K.~Camlibel, and M.~Mesbahi.
\newblock From noisy data to feedback controllers: Nonconservative design via a
  matrix {S}-lemma.
\newblock {\em IEEE Transactions on Automatic Control}, 67(1):162--175, 2022.

\bibitem{bisoffi2020data}
A.~Bisoffi, C.~De Persis, and P.~Tesi.
\newblock Data-based stabilization of unknown bilinear systems with guaranteed
  basin of attraction.
\newblock {\em Systems \& Control Letters}, 145:104788, 2020.

\bibitem{yuan2021data}
Z.~Yuan and J.~Cortés.
\newblock Data-driven optimal control of bilinear systems.
\newblock {\em IEEE Control Systems Letters}, 6:2479--2484, 2022.

\bibitem{guoTAC2021poly}
M.~Guo, C.~De Persis, and P.~Tesi.
\newblock Data-driven stabilization of nonlinear polynomial systems with noisy
  data.
\newblock {\em IEEE Transactions on Automatic Control}, 67(8):4210--4217, 2022.

\bibitem{dai2020semi}
T.~Dai and M.~Sznaier.
\newblock A semi-algebraic optimization approach to data-driven control of
  continuous-time nonlinear systems.
\newblock {\em IEEE Control Systems Letters}, 5(2):487--492, 2020.

\bibitem{Nejati2022poly}
A.~Nejati, B.~Zhong, M.~Caccamo, and M.~Zamani.
\newblock Data-driven controller synthesis of unknown nonlinear polynomial
  systems via control barrier certificates.
\newblock In {\em 4th Annual Learning for Dynamics and Control Conference,
  PMLR}, 2022.

\bibitem{strasser2021data}
Robin Str{\"a}sser, Julian Berberich, and Frank Allg{\"o}wer.
\newblock Data-driven control of nonlinear systems: Beyond polynomial dynamics.
\newblock In {\em 2021 60th IEEE Conference on Decision and Control (CDC)},
  pages 4344--4351. IEEE, 2021.

\bibitem{Verhoek2022lpv}
C.~Verhoek, R.~T\'{o}th, and H.~Abbas.
\newblock Direct data-driven state-feedback control of linear parameter-varying
  systems.
\newblock {\em arXiv:2211.17182}, 2022.

\bibitem{dai2021statedependent}
T.~Dai and M.~Sznaier.
\newblock Nonlinear data-driven control via state-dependent representations.
\newblock In {\em IEEE Conference on Decision and Control}, pages 5765--5770,
  2021.

\bibitem{NonlinearityCancellation2023}
C.~{De Persis}, M.~Rotulo, and P.~Tesi.
\newblock Learning controllers from data via approximate nonlinearity
  cancellation.
\newblock {\em IEEE Transactions on Automatic Control (Early Access)}, pages
  1--16, 2023.

\bibitem{Umlauft2018GP}
J.~Umlauft, L~P\"{o}hler, and S.~Hirche.
\newblock An uncertainty-based control {L}yapunov approach for control-affine
  systems modeled by {G}aussian process.
\newblock {\em IEEE Control Systems Letters}, 2(3):483--488, 2018.

\bibitem{devonport2020bayesian}
A.~Devonport, H.~Yin, and M.~Arcak.
\newblock Bayesian safe learning and control with sum-of-squares analysis and
  polynomial kernels.
\newblock In {\em 2020 59th IEEE Conference on Decision and Control (CDC)},
  pages 3159--3165, 2020.

\bibitem{Fraile2020FL}
L.~Fraile, M.~Marchi, and P.~Tabuada.
\newblock Data-driven stabilization of {SISO} feedback linearizable systems.
\newblock {\em arXiv:2003.14240}, 2020.

\bibitem{luppi2022data}
A.~Luppi, C.~De~Persis, and P.~Tesi.
\newblock On data-driven stabilization of systems with nonlinearities
  satisfying quadratic constraints.
\newblock {\em Systems \& Control Letters}, 163:105206, 2022.

\bibitem{Cheah2022NL}
S.~Cheah, D.~Bhattacharjee, M.~Hemati, and R.~Caverly.
\newblock Robust local stabilization of nonlinear systems with
  controller-dependent norm bounds: A convex approach with input-output
  sampling.
\newblock {\em arXiv:2212.03225}, 2022.

\bibitem{Cetinkaya2021}
A.~Cetinkaya and M.~Kishida.
\newblock Nonlinear data-driven control for stabilizing periodic orbits.
\newblock In {\em 2021 60th IEEE Conference on Decision and Control (CDC)},
  pages 4326--4331, 2021.

\bibitem{Guo2022Taylor}
M.~Guo, C.~{De Persis}, and P.~Tesi.
\newblock Data-driven stabilizer design and closed-loop analysis of general
  nonlinear systems via {T}aylor's expansion.
\newblock {\em arXiv:2209.01071}, 2022.

\bibitem{Martin2022Taylor}
T.~Martin, D.~Sch\"{o}n, and F.~Allg\"{o}wer.
\newblock Gaussian inference for data-driven state-feedback design of nonlinear
  systems.
\newblock {\em arXiv:2212.03225}, 2022.

\bibitem{pillonetto2014kernel}
G.~Pillonetto, F.~Dinuzzo, T.~Chen, G.~De Nicolao, and L.~Ljung.
\newblock Kernel methods in system identification, machine learning and
  function estimation: A survey.
\newblock {\em Automatica}, 50(3):657--682, 2014.

\bibitem{scharnhorst2022robust}
P.~Scharnhorst, E.~Maddalena, Y.~Jiang, and C.~Jones.
\newblock Robust uncertainty bounds in reproducing kernel hilbert spaces: A
  convex optimization approach.
\newblock {\em IEEE Transactions on Automatic Control (Early Access)}, pages
  1--13, 2023.

\bibitem{maddalena2021deterministic}
E.T. Maddalena, P.~Scharnhorst, and C.N. Jones.
\newblock Deterministic error bounds for kernel-based learning techniques under
  bounded noise.
\newblock {\em Automatica}, 134:109896, 2021.

\bibitem{Maddalena2021}
E.~Maddalena, P.~Scharnhorst, Y.~Jiang, and C.~Jones.
\newblock {KPC}: Learning-based model predictive control with deterministic
  guarantees.
\newblock In {\em Learning for Dynamics and Control. PMLR}, pages 1015--1026,
  2021.

\bibitem{van2022kernel}
H.J van Waarde and R.~Sepulchre.
\newblock Kernel-based models for system analysis.
\newblock {\em IEEE Transactions on Automatic Control}, 2022.

\bibitem{huang2022robust}
L.~Huang, J.~Lygeros, and F.~D{\"o}rfler.
\newblock Robust and kernelized data-enabled predictive control for nonlinear
  systems.
\newblock {\em arXiv:2206.01866}, 2022.

\bibitem{wendland2004scattered}
H.~Wendland.
\newblock {\em Scattered data approximation}, volume~17.
\newblock Cambridge university press, 2004.

\bibitem{aronszajn1950theory}
Nachman Aronszajn.
\newblock Theory of reproducing kernels.
\newblock {\em Transactions of the American mathematical society},
  68(3):337--404, 1950.

\bibitem{scholkopf2001generalized}
B.~Sch{\"o}lkopf, R.~Herbrich, and A.~Smola.
\newblock A generalized representer theorem.
\newblock In {\em Computational Learning Theory: 14th Annual Conference on
  Computational Learning Theory, COLT 2001 and 5th European Conference on
  Computational Learning Theory, EuroCOLT 2001 Amsterdam, The Netherlands, July
  16--19, 2001 Proceedings 14}, pages 416--426. Springer, 2001.

\bibitem{fasshauer2011positive}
G.~Fasshauer.
\newblock Positive definite kernels: past, present and future.
\newblock {\em Dolomites Research Notes on Approximation}, 4:21--63, 2011.

\end{thebibliography}

\end{document}